\documentclass{article}

\usepackage{amsmath}
\usepackage{amssymb}
\usepackage{a4}
\usepackage{algorithmic}
\usepackage{algorithm}
\usepackage{graphicx}
\usepackage{hyperref}

\def \R{{\mathbb R}}

\def \N{{\mathbb N}}

\DeclareMathOperator{\sign}{sign}
\DeclareMathOperator{\Tr}{Tr}
\DeclareMathOperator{\Id}{Id}

\newtheorem{theorem}{Theorem}[section]

\newtheorem{lemma}[theorem]{Lemma}

\newtheorem{definition}[theorem]{Definition}

\newenvironment{proof}
 {\begin{trivlist} \item[\hskip \labelsep {\bf Proof.}]}
 {\hfill$\Box$\end{trivlist}}

\begin{document}

\title{Compressed Manifold Modes: Fast Calculation and Natural Ordering}
\author{Kevin Houston \\
School of Mathematics \\ University of Leeds \\ Leeds, LS2 9JT, U.K. \\
e-mail: k.houston@leeds.ac.uk \\
www.maths.leeds.ac.uk/$\sim $khouston/
}
\date{\today }
\maketitle
\begin{abstract}
Compressed manifold modes are locally supported analogues of eigenfunctions of the Laplace-Beltrami operator of a manifold. In this paper we describe an algorithm for the calculation of modes for discrete manifolds that, in experiments, requires on average 47\% fewer iterations and 44\% less time than the previous algorithm. We show how to naturally order the modes in an analogous way to eigenfunctions, that is we define a compressed eigenvalue. Furthermore, in contrast to the previous algorithm we permit unlumped mass matrices for the operator and we show, unlike the case of eigenfunctions, that modes can, in general, be oriented.
\end{abstract}

\section{Introduction}
The eigenvalues and eigenfunctions of the Laplace-Beltrami operator (LBO) on a discrete manifold have found many applications in geometry processing, for example, in shape matching, remeshing (such as quadrangulation), smoothing, and shape identification, see for example, \cite{Crane:2013:DGP, LevyZhang}.

In analogy with a Fourier basis, the eigenfunctions of the LBO form a basis, called the {\textbf{manifold harmonic basis}} (see \cite{vallet2008spectral}), of the functions on the manifold. 
One major drawback of this basis is that it does not relate, at least in an intuitive way, to observable features of manifolds. In the pioneering work \cite{ozolicnvs2013compressed,Ozolins04022014} it is shown that one can produce on $\R ^n$ a set of functions with localized support, i.e., compactly supported, which will function as a basis. They named these {\textbf{compressed modes}}. They proposed that these may be used as a natural basis for solving PDEs and suggested that they could be extended to general discrete manifolds. This generalization was made in  \cite{CMM14} and were called {\textbf{compressed manifold modes}} (CMMs).
These CMMs are locally supported and, crucially, they seem to be supported at important features of the manifold. For example in the top line of Figure~\ref{fig:locsup} one can see that the collection of six CMMs is supported on the arms, legs, head and lower torso. That is, features that a human would identify are detected by the modes. The six modes pictured on the less symmetric Aquarius the Water Carrier, are also supported on regions which a human may identify as regions of interest.
Twenty modes are given later in Figure~\ref{fig:teapot} for the classic teapot.

\begin{figure}
\begin{center}
\includegraphics[width=12cm]{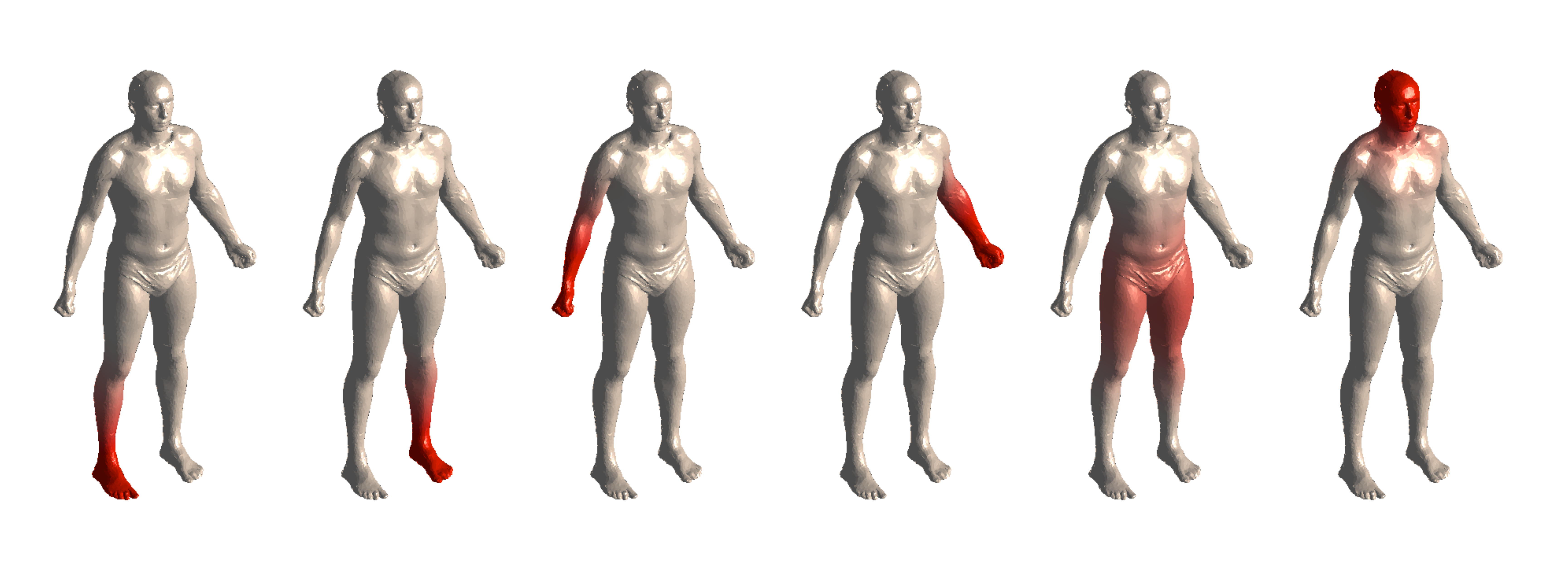}

\includegraphics[width=12cm]{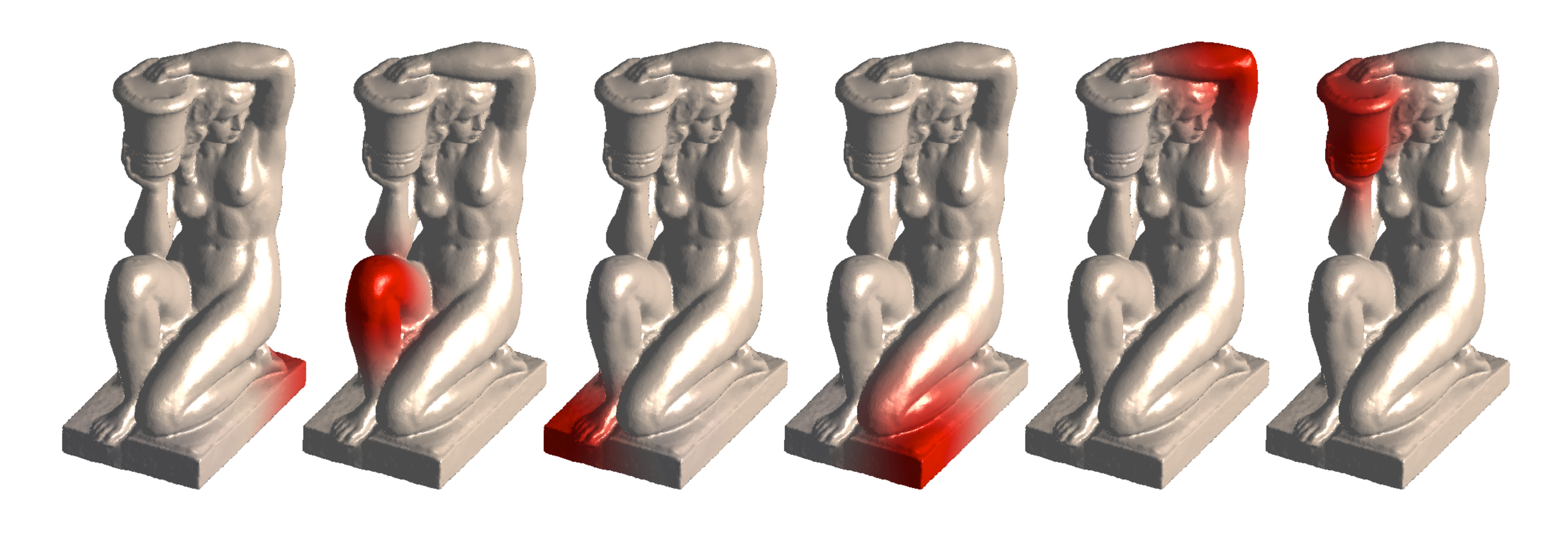}
\caption{\label{fig:locsup}Local support and identification of natural features.}
\end{center}
\end{figure}
Similarly the support for $15$ modes on a L-shaped mesh is local and identifies significant features such as corners as shown in Figure~\ref{fig:lshape}.
\begin{figure}
\begin{center}
\includegraphics[width=12cm]{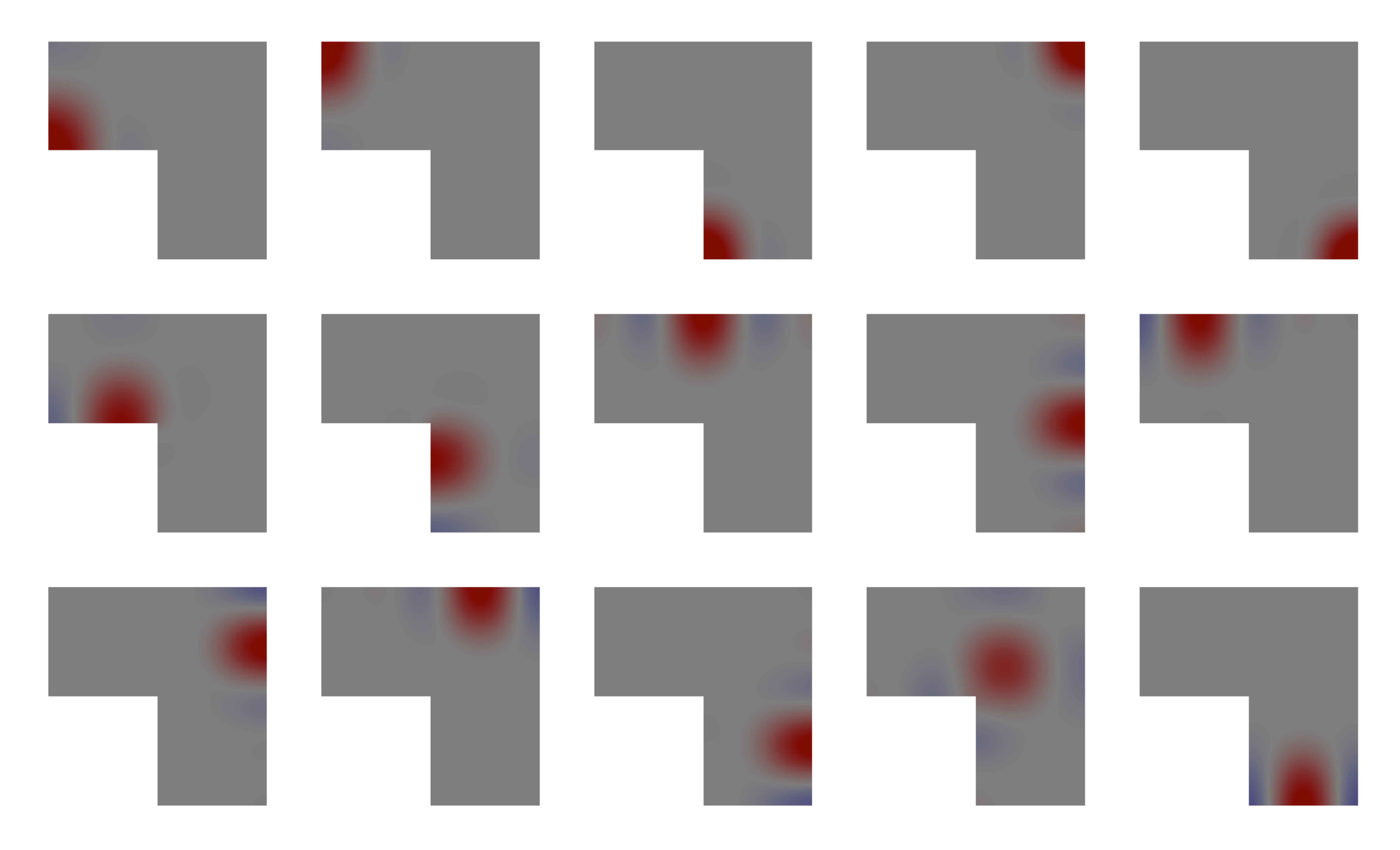}
\caption{\label{fig:lshape}Support on a L-shaped mesh, red is positive and blue is negative.}
\end{center}
\end{figure}

In \cite{CMM14} the authors show that, in contrast to the eigenfunctions and eigenvalues of the LBO, the CMMs are robust with respect to noise and holes. Examples are given where shape matching is achieved in the presence of noise and partial sampling of meshes. Furthermore, the theory in \cite{ozolicnvs2013compressed,Ozolins04022014} is developed as a way of solving partial differential equations.  Hence, CMMs are potentially important objects of study in a diverse range of subjects.

Two significant problems of compressed manifold modes that are identified in \cite{CMM14} are speed and ordering. They take longer to generate than eigenfunctions and no method for ordering was given. Further problems are that the algorithm in \cite{CMM14} is restricted to an LBO with diagonal mass matrix and it is noted that there is ambiguity of sign of the modes in the same way that unit eigenfunctions of a matrix are only defined up to sign.

This paper deals with all these problems. The main contributions are as follows: 
\begin{enumerate}
\item We adapt an algorithm from \cite{goldstein2012fast} to speed up the algorithm given in \cite{CMM14} for computation of the modes. This accelerated algorithm is described in Section~\ref{sec:acmm} and experiments show an average 44\% reduction in the time taken with a 47\% decrease in the number of iterations required.
The accuracy of the accelerated version is measured in Section~\ref{sec:accuracy}.
\item We describe a natural ordering for the modes in Section~\ref{sec:order} that is entirely analogous to the ordering of eigenvalues of a matrix.
\item The new algorithm allows non-diagonal mass matrices. The explanation of this is given in Section~\ref{sec:lump}.
\item Section~\ref{sec:flip} shows how, in contrast to the case of eigenvectors, the modes can be oriented, i.e., can be given a sign. This can be seen by comparing the colours in Figure~\ref{fig:cmm14fig13} with that of Figure~13 in \cite{CMM14}. The modes in Figure~\ref{fig:cmm14fig13} have been `flipped' where necessary.
\end{enumerate}

The Matlab code of the implementation of the algorithms is freely available under a Creative Commons 
Attribution-NonCommercial-ShareAlike (CC BY-NC-SA) license. It can be found at 
\url{http://www1.maths.leeds.ac.uk/~khouston/}.

Acknowledgements: My thanks to Thomas Neumann for helpful discussions and comments.
The meshes of humans come from the SCAPE dataset, \cite{SCAPE} and the Aquarius mesh is from the EPFL Computer Graphics and Geometry Laboratory.
The Algorithms Latex bundle by Rog\'erio Brito was used in preparation of the manuscript.

%
%

\section{Background, previous and related work}

In recent years there has been much interest in the eigenvalues and eigenfunctions of the discrete Laplace-Beltrami operator. For a smooth manifold $M$ the classical differential geometry Laplace-Beltrami operator, denoted $\Delta $, has a set of eigenfunctions $\{ \varphi _k \}$ and associated eigenvalues, $\{ \lambda _k\} $, determined by 
\[
\Delta \varphi _k = \lambda _k \varphi _k
\]
where $k\in \N $ and $\lambda _k\in \R$. See \cite{berger2003panoramic}. 
The self-adjointness of $\Delta $ implies that the eigenvalues are real and that the eigenfunctions are orthogonal with respect to the $L_2$-inner product: $\langle f, g\rangle = \int_M fg $.

In the discrete case we denote the LBO by $L$ and decompose it as $L=A^{-1}W$ where $A$ is a {\textbf{mass matrix}}, i.e., a matrix which relates to the area/volume around the vertices of the discrete manifold, and $W$ is a weight matrix, such as the cotan matrix, (see \cite{Crane:2013:DGP,LevyZhang} for further references). 

Both $A$ and $W$ are symmetric $N\times N$ matrices where $N$ is the number of vertices.
The matrix $L=A^{-1}W$ is in general not symmetric but is self-adjoint with respect to the $L_2$-inner product given by $\langle v, w \rangle =v^TAw$, where $v^T$ denotes the transpose of the vector $v$. This implies that its eigenvalues are real and its eigenvectors are orthogonal with respect to the inner product.
The eigenvalues and eigenvectors of $L$ are solutions of the symmetric eigenvalue problem 
\[
W\varphi = \lambda A\varphi .
\]
In this paper we write $K$ eigenfunctions as the columns of the $N\times K$ matrix $\Phi $.

In \cite{ozolicnvs2013compressed} {\textbf{compressed modes}} are introduced as a solution of a minimization problem involving the Hamiltonian operator $H=-(1/2)\Delta +V(x)$, where $V$ is a potential function, %
with what is called an $L_1$-regularizing term, that is, $(1/\mu)\sum_j |\psi _j|_1$ for modes $\psi _j$ where $\mu \in \R $ is a parameter controlling how localized the modes are. 

We produce a set $\Psi _K=\{ \psi _j \} _{j=1}^K$ arising from the variational problem 
\[
\min_{\Psi _K} \sum_{j=1}^K \left(  \dfrac{1}{\mu } \left| \psi _j \right|_1 + \langle \psi _j , H\psi _j \rangle  \right) 
{\text{ such that }} \Psi ^T\Psi = \Id
\] 
where $\Psi $ is the $N\times K$ matrix given by arranging the $\psi _j$ in columns.

The generalization, called {\textbf{compressed manifold modes}}, described in \cite{CMM14} is as follows. 
For $K\in \N$ and $\mu \in \R $, the first $K$ compressed manifold modes are the columns of the matrix $\Phi $, where $\Phi $ is determined by the constrained optimization problem
\[
\min_{\Phi } \Tr \left( \Phi ^T W \Phi \right) + \mu || \Phi ||_1
{\text{ such that }} \Phi ^TA \Phi = \Id.
\]
Here $||\Phi ||_1$ is the sum of the absolute values of entries of $\Phi $. (Note that, rather confusingly, the compression parameter is $1/\mu $ in \cite{ozolicnvs2013compressed} and $\mu $ in \cite{CMM14}. We will follow the latter as our algorithm is a generalization of theirs.)
A mode is an analogue of the Laplace-Beltrami eigenfunctions.

The parameter $\mu $ is a measurement of the compression of the support of the modes. A large $\mu $ means large compression of support, i.e, the support gets smaller, and a small $\mu $ has small compression, i.e., large support.

%
%
\section{Accelerated ADMM Algorithm}
\label{sec:acmm}
This section contains the first contribution of the paper -- an accelerated version of the algorithm presented in \cite{CMM14} to calculate the compressed manifold modes. In the tests this new algorithm required 47\% fewer iterations to reach convergence and resulted in a 44\% time improvement.

In \cite{ozolicnvs2013compressed}, Ozoli{\c{n}}{\v{s}} et al propose a Splitting Orthogonality Constraint (SOC) algorithm (described in detail in \cite{split2014}) for calculating compressed modes in the case that the manifold is $\R ^n$. In \cite{CMM14}, Neumann et al show that the SOC algorithm does not function well for more general manifolds and to counteract this they propose an Alternating Direction Method of Multipliers (ADMM) algorithm. The ADMM method is known to be particularly slow, see \cite{BoydADMM}. A method of acceleration based on the Nesterov method, \cite{nesterov1983method}, is described in \cite{goldstein2012fast} and we propose a modification of this to increase the speed of the algorithm from \cite{CMM14}. The algorithm is a variant of the gradient descent method with an over-relaxation step. 

Let $f$ and $g$ be functions and that we wish to minimize $f(u)+g(v)$ subject to $Au+Bv=c$. Let $\rho $ be a penalty parameter as in \cite{BoydADMM}. The accelerated algorithm is given in Algorithm~\ref{alg:cmms}. 
\begin{algorithm}
\caption{Accelerated Calculation of CMMs.}
\label{alg:cmms}
\begin{algorithmic}[1]
\REQUIRE $\alpha _1=1$, $\eta=0.999$.
\FOR{$k=1,2,3,\dots $}
\STATE $u_k=\arg\min f(u) + \dfrac{\rho}{2} ||Au+B\widehat{v}_k+c-\widehat{\lambda }_k ||^2$
\STATE $v_k=\arg\min g(v) + \dfrac{\rho}{2} ||Au_k+Bv+c-\widehat{\lambda }_k||^2$
\STATE $\lambda _k =\widehat{\lambda} _k  +Au_k+Bv_k+c$
\STATE $c_k=\rho \left(||\lambda _k -\widehat{\lambda}_k ||^2 + ||B(v_k-\widehat{v} _k)||^2\right) $ 
\IF{$c_k<\eta c_{k-1}$}
\STATE $\alpha _k=1$
\ENDIF
\STATE $\alpha _{k+1} = \dfrac{1+\sqrt{1+4 \alpha_k^2}}{2}$
\STATE $\widehat{v}_{k+1} = v_k + \dfrac{\alpha _k -1}{\alpha _{k+1}} \left( v_k - v_{k-1} \right) $
\STATE $\widehat{\lambda }_{k+1} = \lambda _k + \dfrac{\alpha _k -1}{\alpha _{k+1}} \left( \lambda _k - \lambda _{k-1} \right) $
\ENDFOR
\end{algorithmic}
\end{algorithm}
To compute the compressed manifold modes we proceed as described in Section 3.2 of \cite{CMM14}. The optimization function is split into the sum of three (rather than two) functions:
\[
\Tr (\Phi ^TL\Phi ) + \mu || \Phi || _1 + \iota (\Phi ) 
\]
where the indicator function $\iota$ is defined by
\[
\iota (\Phi ) = \left\{
\begin{array}{cl}
0, & {\text{if }} \Phi ^T A \Phi =\Id, \\
\infty , & {\text{otherwise}}.
\end{array}
\right.
\]
They then reformulate the problem as 
\[
\min_{\Phi , S, E} \iota (\Phi ) + \Tr (E^T LE) + \mu || S||_1 
{\text{ such that }} \Phi =S, ~\Phi =E.
\]
We proceed similarly to \cite{CMM14} and use the following in Algorithm~1 : $f=\iota$, $u=\Phi$, $v=\left[ \begin{array}{c}E\\S \end{array} \right] $ and
\[
g \left( \left[ \begin{array}{c}E\\S \end{array} \right]
\right) = 
\left[ \begin{array}{c}\Tr (E^TLE) \\\mu ||S||_1 \end{array} \right] .
\]
We set 
\[
A=\left[
\begin{array}{c}
I \\
I 
\end{array}
\right], 
\qquad 
B=\left[
\begin{array}{cc}
-I & 0  \\
0 & -I 
\end{array}
\right]
\quad
\text{ and } \quad
c=0.
\]

Note that there is no guarantee that this algorithm converges. That the algorithm in \cite{goldstein2012fast} converges for a weakly convex optimization problem is shown in \cite{goldstein2012fast}. The essential part of the proof is the monotonic decrease in a certain residual. The optimization problem we wish to solve is not weakly convex and the residual does not monotonically decrease (as can be seen in examples). Hence we make a modification to the restart rule to ensure that the algorithm does not become stuck on the same values.
The acceleration process requires a parameter, denoted $\eta $, to initiate the restart. In all the experiments in this paper $\eta $ was set to $0.999$.

In the implementation used for the experiments in this paper (and in \cite{CMM14}) the penalty parameter is varied as in Section 3.4 of \cite{BoydADMM}. This greatly increases the speed of both algorithms.

Furthermore, convergence was determined by the use of the {\textbf{primal and dual residuals}}. 
The primal residual at iteration $k$ is $r_k=Au_k+Bv_k-c$ and the dual residual is $s_k=\rho A^TB(v_k-v_{k-1})$. 
The precise condition for stopping the algorithm are described in Section~3.3.1 of \cite{BoydADMM}.
The algorithm terminates when $||r_k||_2\leq \epsilon^{\text{pri}}$ and $||s_k||_2\leq  \epsilon ^{\text{dual}}$ for the two tolerances, $\epsilon^{\text{pri}}$ and $\epsilon ^{\text{dual}}$. These are determined from
\begin{align*}
\epsilon^{\text{pri}} &=\sqrt{2n} \epsilon^{\text{abs}} + \epsilon^{\text{rel}} \max \left\{ ||\Phi ||_2 , \sqrt{||E||_2^2+||S||_2^2} \right\} , \\
\epsilon ^{\text{dual}}&=\sqrt{n} \epsilon^{\text{abs}} + \epsilon^{\text{rel}} ||\lambda ||_2
\end{align*}
where $n$ is the number of vertices for the mesh, $\epsilon^{\text{abs}}$ is an absolute tolerance, $\epsilon^{\text{rel}}$ is a relative tolerance and $\lambda $ is the dual variable in the algorithm. In all experiments, following \cite{CMM14}, we took 
$\epsilon^{\text{abs}}=10^{-8}$ and $\epsilon^{\text{rel}}=10^{-6}$.

In the experiments the two algorithms were given the same random initialization of numbers between $0$ and $1$. This was repeated 30 times for each mesh. 
Table~\ref{tab:compare} compares the average number of iterations and average time taken for the accelerated algorithm described above and the one from \cite{CMM14}. The meshes Stand, Crouch and Bent Limbs are the those from first, second and third lines of Figure~\ref{fig:cmm14fig13} respectively; L-shape is the L-shape from Figure~\ref{fig:lshape}; Aquarius is the water carrier statue mesh from Figure~\ref{fig:locsup}. The experiments were performed on an iMac with 3.4 GHz Intel core i7 and 8GB RAM.

In the experiments the accelerated algorithm required on average 47\% fewer iterations than the original algorithm with a range between 9\% and 88\%. The improvement in speed was nearly as good with an average 44\% improvement. The range of improvements was 2\% to 87\%. Hence, we can conclude that the improvements are significant and the algorithm is worth implementing when compared to the original.

Typical graphs comparing the sums of the prime and dual residuals for the two algorithms are shown in 
Figure~\ref{fig:compareresidual}. These were selected to give some insight into how the accelerated version behaves. We can see that its residual sum is a condensed and bumpier version of the unaccelerated version's.
The combined residual for the accelerated algorithm looks like it has greater fluctuations and this is indeed the case. This is pictured in close up Figure~\ref{fig:ripples} where one can see that these are Nesterov `ripples' -- a common feature of Nesterov methods. The algorithm restarts at the top of the bounce.

\begin{table}
\begin{center}
\begin{tabular}{|c|c|c|c|c|c|c|c|c|c|c|c|c|c|c|c|}
\hline
\parbox[t]{1cm}{\ \\ \ \\ Mesh}	& \parbox[t]{1cm}{\ \\ \ \\  $\mu $ } 	&	\parbox[t]{1cm}{\ \\ \ \\ $K$}	  	&	\parbox[t]{1cm}{\ \\Fast\\ADMM\\ iter's}	&	\parbox[t]{1cm}{\ \\ ADMM\\ iter's}	&	\parbox[t]{1cm}{\ \\Percent\ reduction}	&	\parbox[t]{1cm}{Fast\\ADMM\\ time\\ (in sec) }	&	\parbox[t]{1cm}{\ \\ ADMM\\ time\\ (in sec)}	&\parbox[t]{1cm}{\ \\Percent\ reduction}	
 \\
\hline
Stand	&	0.008	&	10	&	1381	&	2889	&	48	&	43	&	86	&	45	\\
Crouch	&	0.008	&	10	&	1140	&	3177	&	57	&	38	&	89	&	49	\\
Bent	limbs&	0.008	&	10	&	1218	&	2373	&	49	&	38	&	71	&	47	\\
\hline
Stand&	0.001	&	6	&	1329	&	1980	&	33	&	37	&	53	&	30	\\
Crouch	&	0.001	&	6	&	1916	&	2704	&	30	&	48	&	66	&	27	\\
Bent limbs	&	0.001	&	6	&	1813	&	3640	&	46	&	50	&	95	&	43	\\
\hline
Stand	&	0.0075	&	6	&	829	&	1498	&	44	&	23	&	40	&	42	\\
Crouch	&	0.0075	&	6	&	902	&	1660	&	44	&	23	&	40	&	41	\\
Bent limbs	&	0.0075	&	6	&	922	&	1645	&	43	&	25	&	44	&	41	\\
\hline
L Shape	&	0.02	&	10	&	7867	&	18597	&	57	&	229	&	469	&	51	\\
\hline
Aquarius	&	0.0075	&	6	&	1568	&	4349	&	60	&	243	&	661	&	59	\\
Aquarius	&	0.0075	&	10	&	2220	&	5182	&	55	&	412	&	875	&	50	\\
\hline
Average	&		&		&		&		&	47	&		&		&	44	\\
\hline
\end{tabular}
\end{center}
\caption{\label{tab:compare}Comparison of the accelerated and original ADMM algorithms. For each mesh the average of 30 iterations is taken.}
\end{table}

\begin{figure}
\begin{center}
\includegraphics[width=6cm]{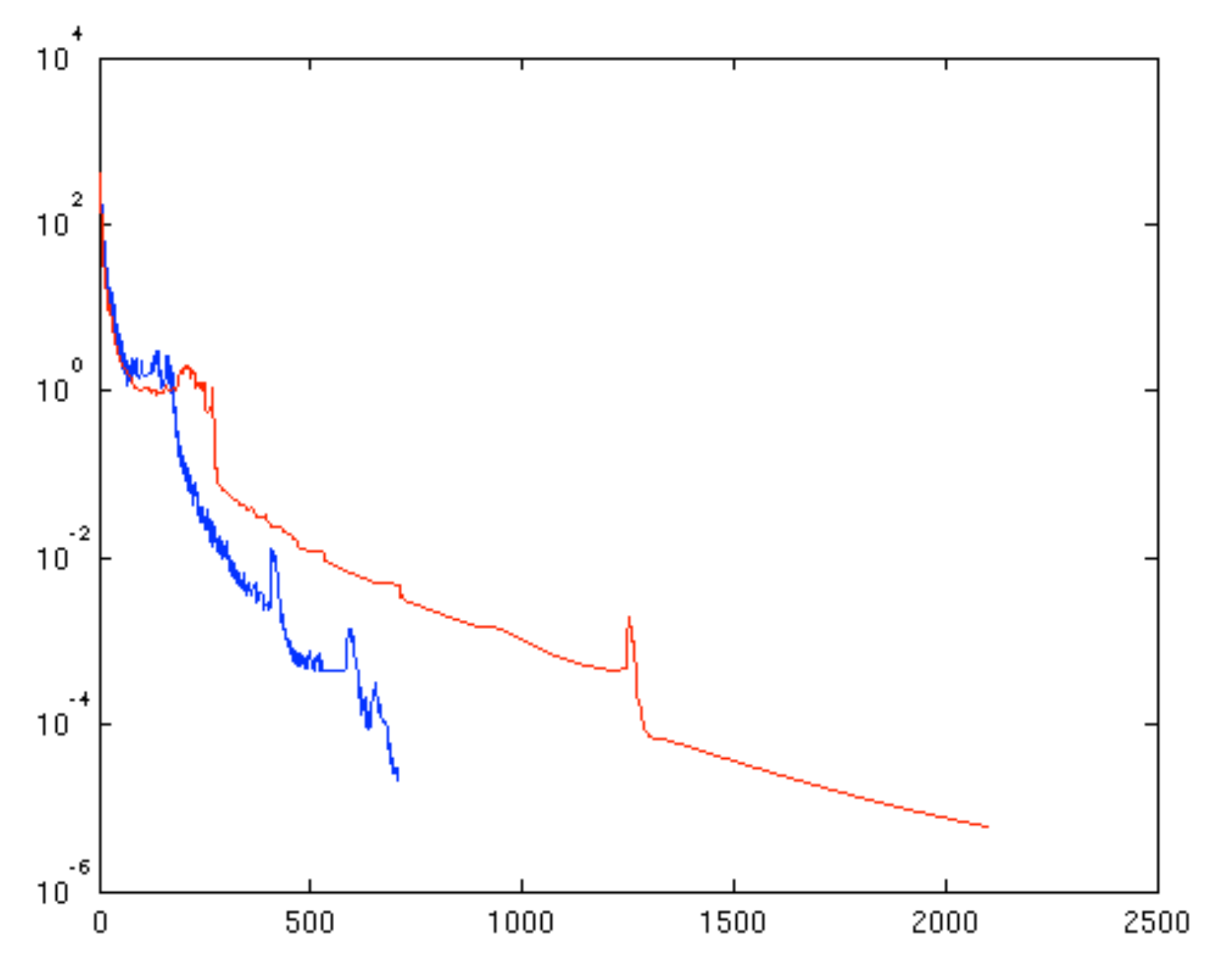}
\includegraphics[width=6cm]{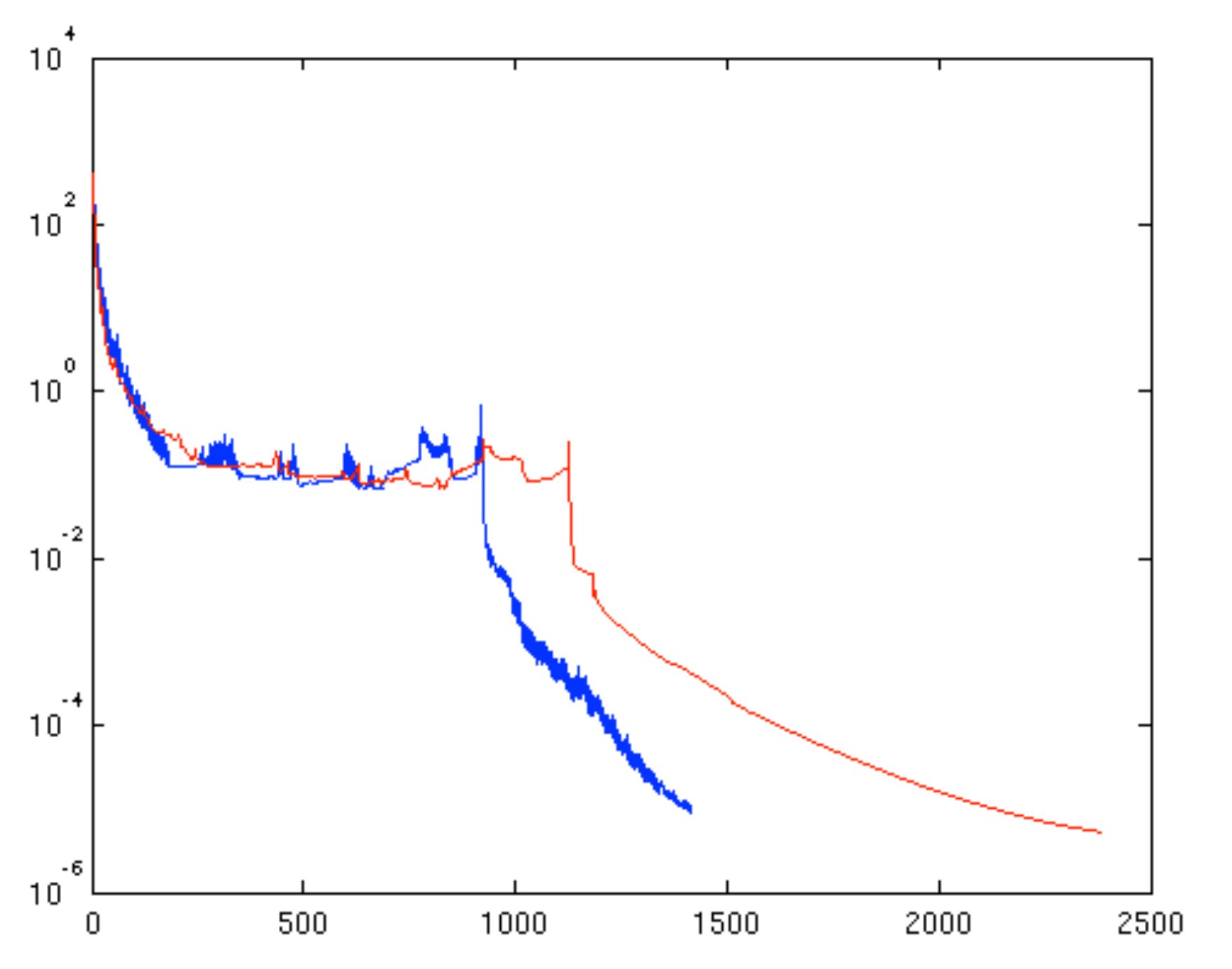}

Stand, $\mu=0.0075$, $K=6$ \quad \quad \quad \quad \quad \quad
Stand, $\mu=0.001$, $K=6$

\includegraphics[width=6cm]{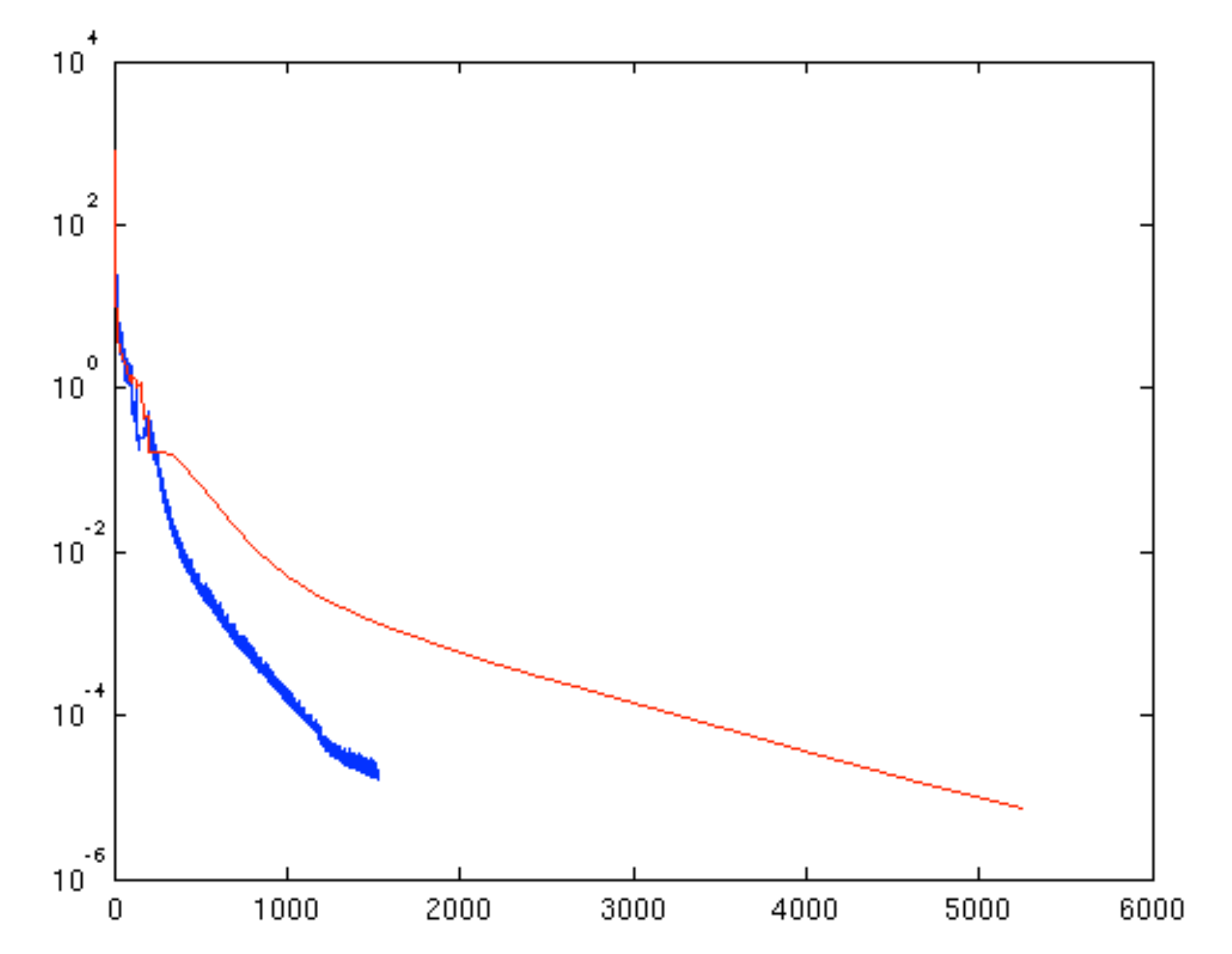}
\includegraphics[width=6cm]{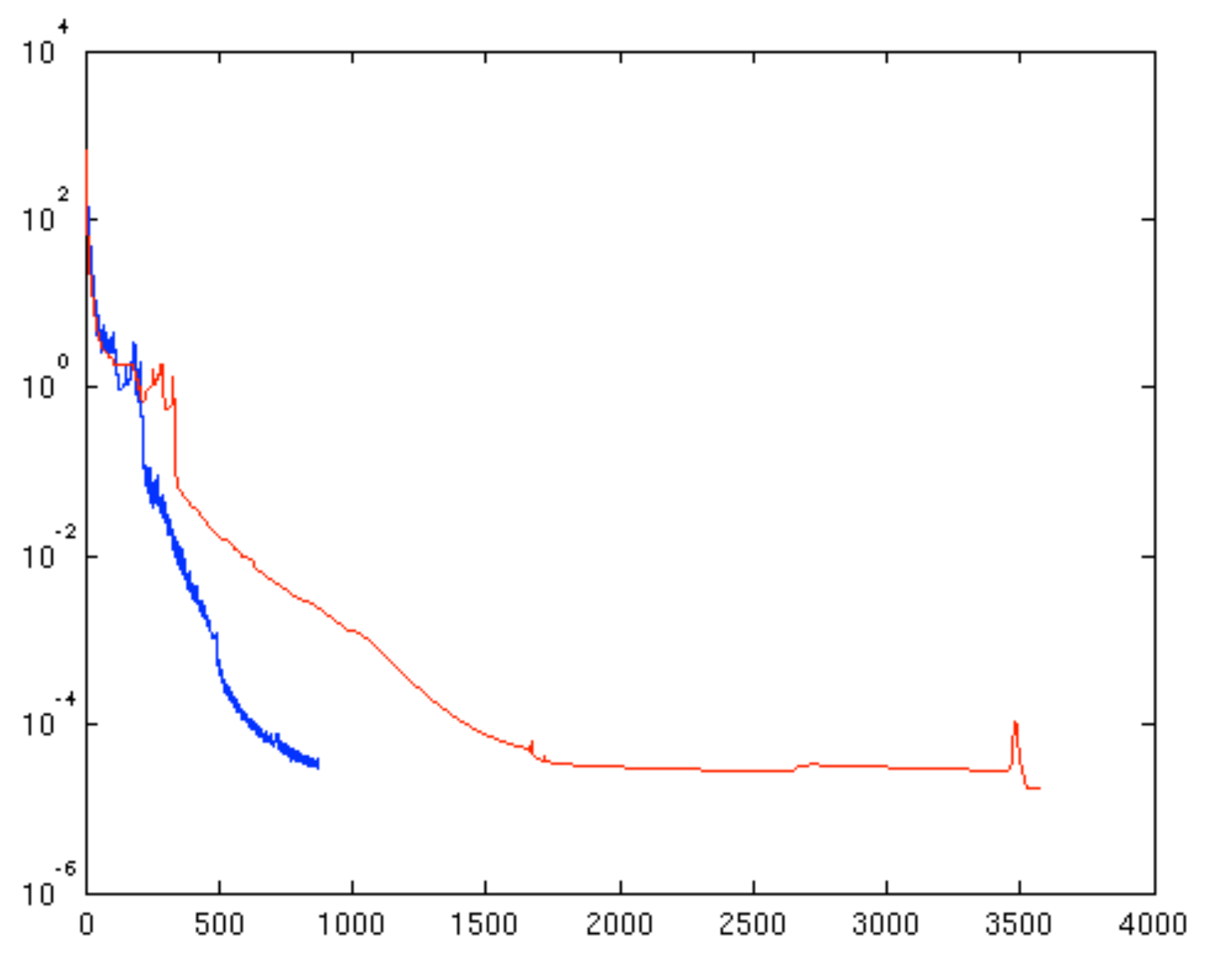}

Aquarius, $\mu=0.0075$, $K=6$ \quad \quad \quad \quad \quad \quad
Stand, $\mu=0.0075$, $K=10$

\includegraphics[width=6cm]{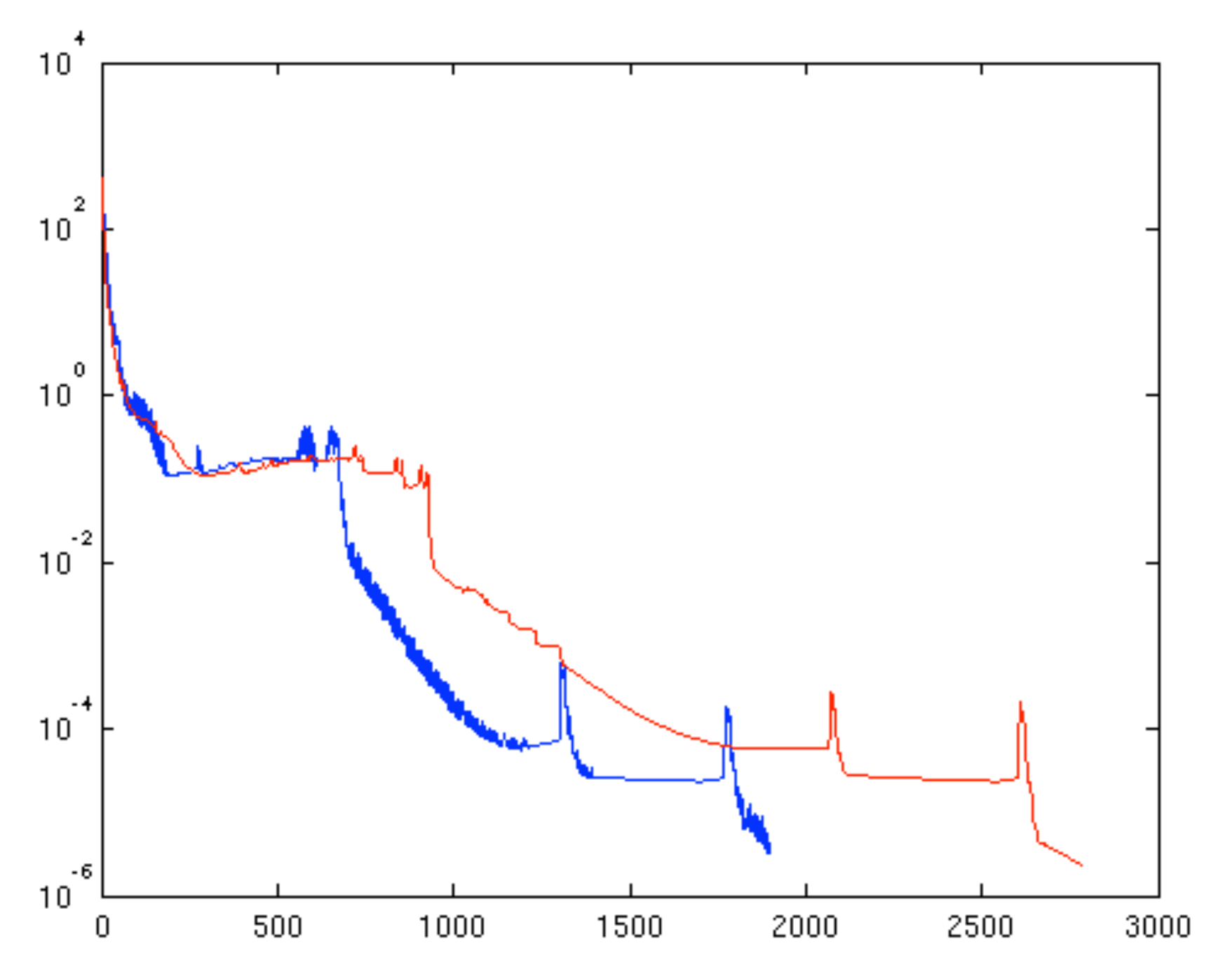}
\includegraphics[width=6cm]{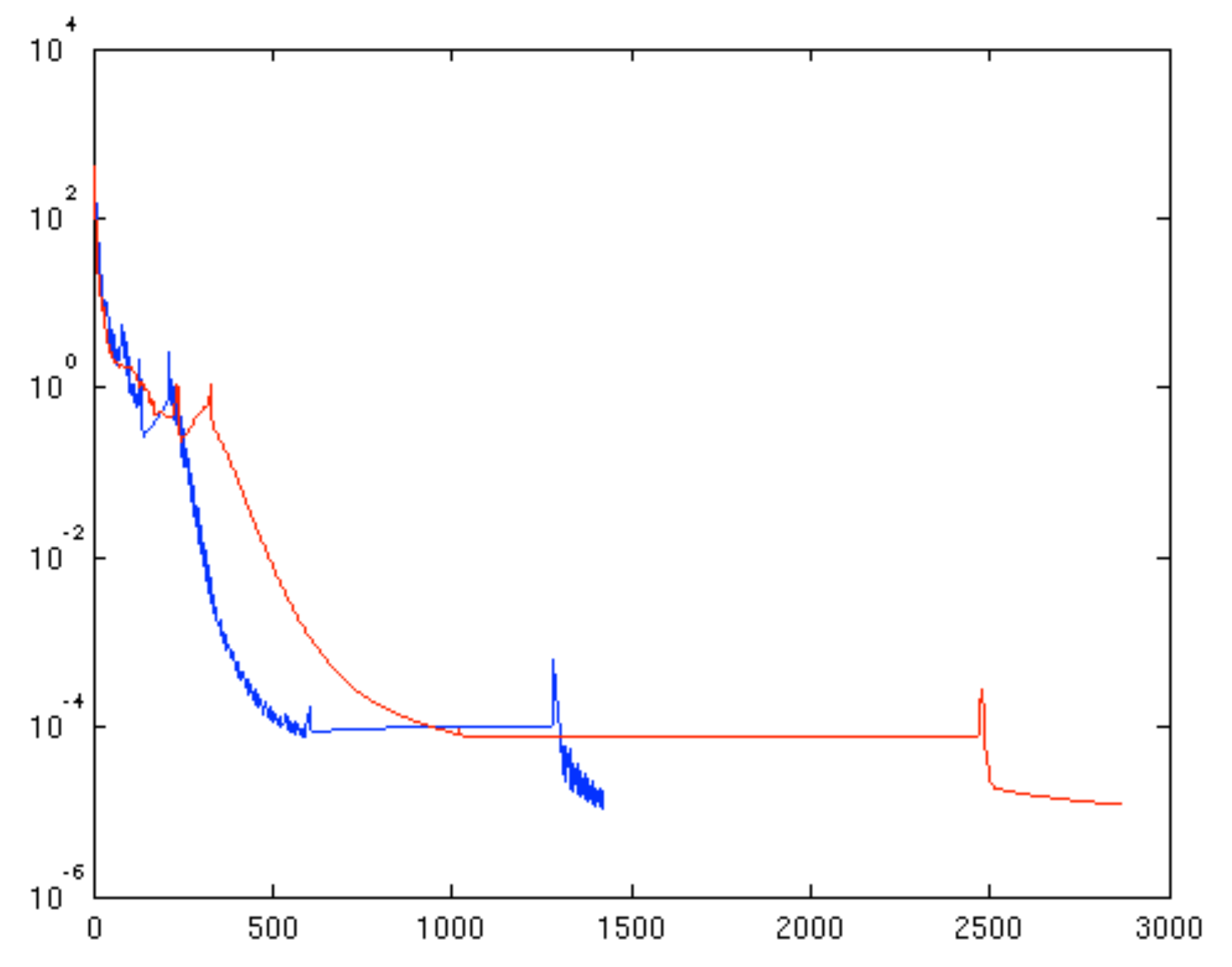}

Crouch, $\mu=0.001$, $K=6$ \quad \quad \quad \quad \quad \quad
Crouch, $\mu=0.0075$, $K=6$

\includegraphics[width=6cm]{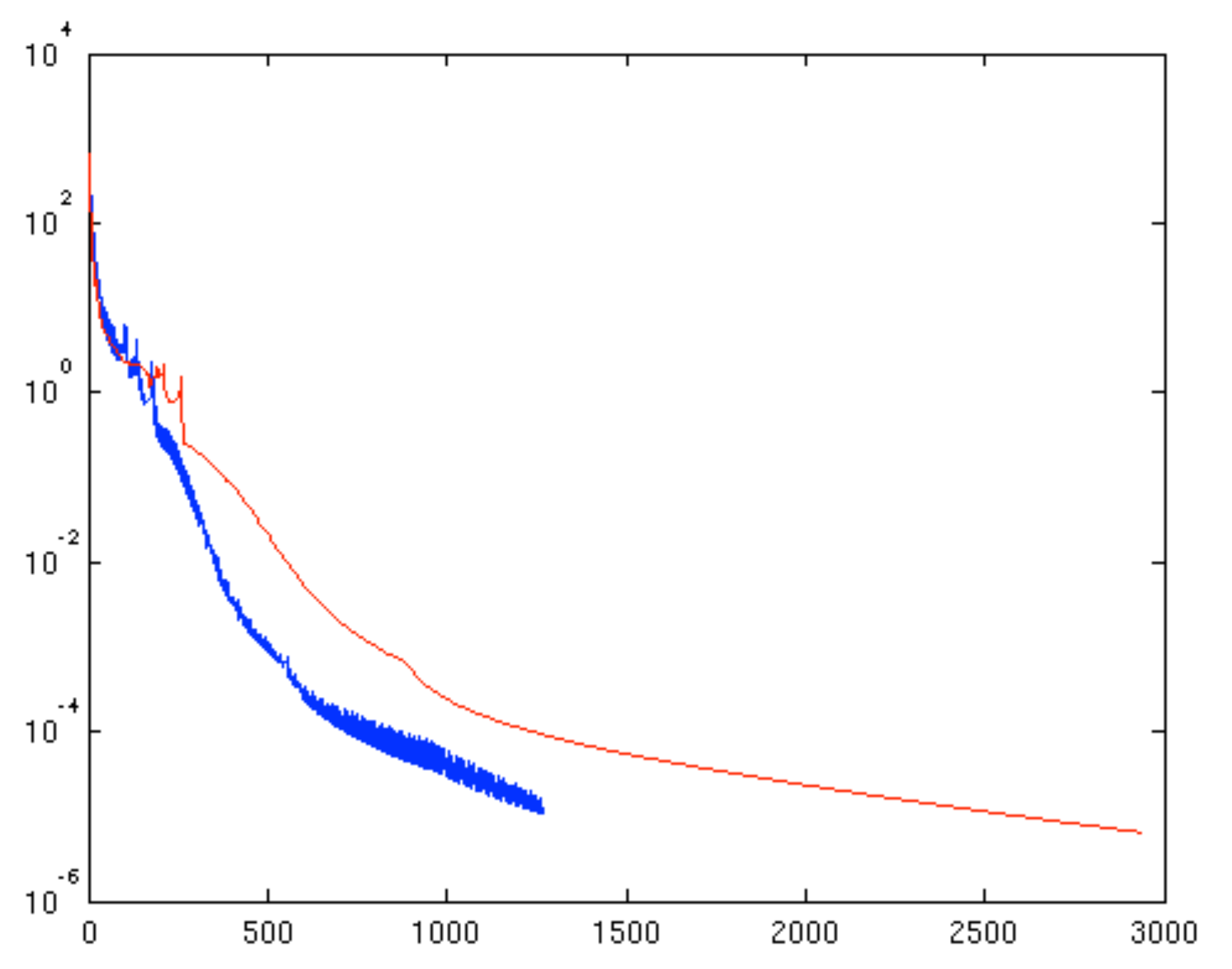}

Crouch, $\mu=0.0075$, $K=10$

\caption{\label{fig:compareresidual}Comparison of typical sum of prime and dual residual for accelerated ADMM (blue) and ADMM (red).}
\end{center}
\end{figure}

\begin{figure}
\begin{center}
\includegraphics[width=14cm]{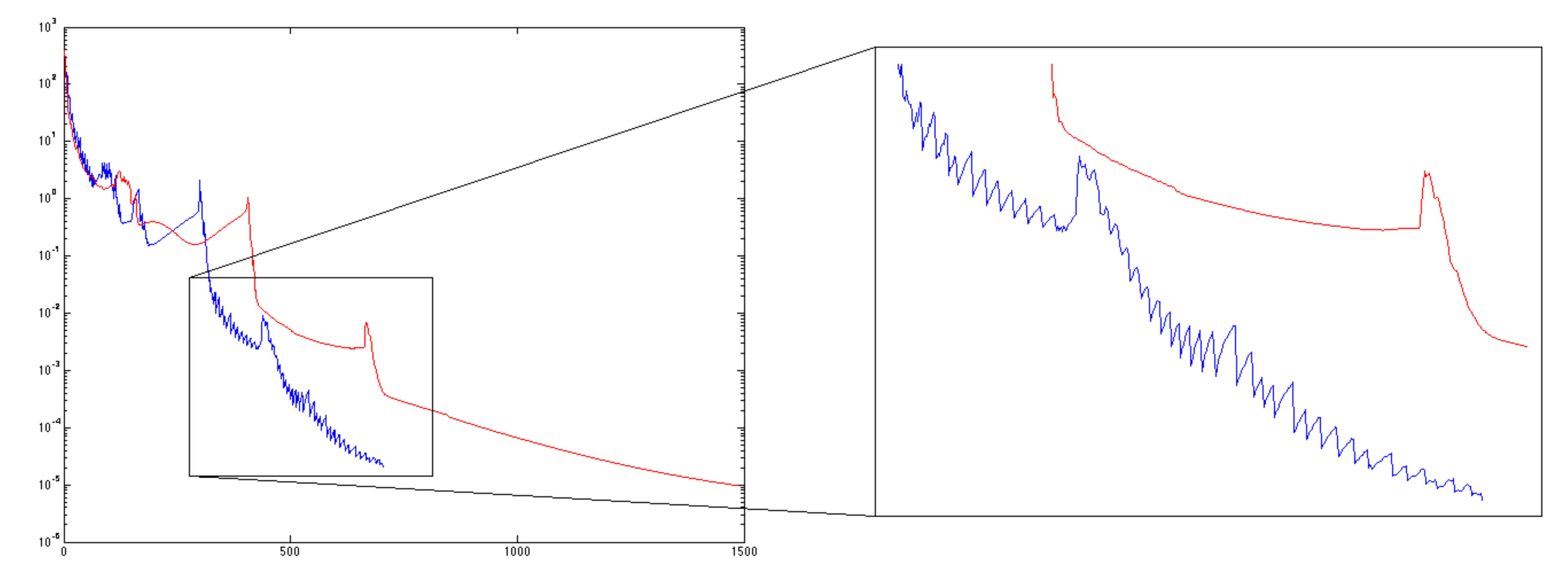}
\caption{\label{fig:ripples}A close up (with scale) of the Nesterov ripples. Mesh was Stand with $\mu=0.008$ and $K=6$.}
\end{center}
\end{figure}

\section{A natural order for compressed manifold modes}
\label{sec:order}
For the Laplace-Beltrami operator we can naturally order the eigenfunctions by their associated eigenvalues. We shall now describe the correct analogous natural ordering for compressed manifold modes. One could naively order the modes according to their Dirichlet energy. That is, ignore the additional $L_1$ term in the minimisation. This leads to poor results as one can see by comparing Figure~\ref{fig:cmm14fig13} (our ordering) and Figure~\ref{fig:dir-ordered} (Dirichlet ordering). The former is better than the latter for these near-isometric surfaces. 

In Figure~\ref{fig:cmm14fig13} the modes for a mesh fall naturally into pairs of consecutive modes with the exception of the third row where one could argue that positions 4 and 5 should be exchanged. One could also argue that the head- and torso-supported modes are not in the same sequence in the three meshes. However, they do only occur in positions 7 and 8. Contrast these small differences with the larger ones in Figure~\ref{fig:dir-ordered}, the Dirichlet ordering. Here the head-supported mode appears in positions 3, 4, and 8. The torso supported mode appears in positions 5 and 6. Furthermore, no mesh consistently has what we could call natural pairs and in particular the bottom mesh has most of the modes not occurring in consecutive pairs. 

Having shown the superiority of the new ordering, let us now describe it in detail. Let $L=A^{-1}W$ denote the Laplace-Beltrami operator where $A$ is the mass matrix and $W$ the weight matrix. Suppose we wish to find $K$ modes. Then the modes can be represented as the $K$ columns of $\Phi $ where $\Phi $ is determined by the constrained optimization problem
\[
\min_{\Phi } \Tr \left( \Phi ^T W \Phi \right) + \mu || \Phi ||_1
{\text{ such that }} \Phi ^TA \Phi = \Id.
\]

Since $A$ is symmetric the constraint $\Phi^TA \Phi=\Id$ determines $K(K-1)/2$ distinct equations. If we apply the Lagrange multiplier method we require $K(K-1)/2$ Lagrange multipliers $\lambda _{ij}$, with $1\leq i \leq k$ and $i\leq j \leq k$.

\begin{lemma}
The Lagrangian function, $\mathcal{L}$, for this constrained optimization problem is 
\begin{equation*}
\mathcal{L}(\Phi) =  \Tr  \left( \Phi^T W \Phi \right) + \mu || \Phi||_1 -\Tr  \left(\left( \Phi^TA \Phi- \Id \right)  \Lambda \right) 
\end{equation*}
where $\Lambda $ is the real symmetric $K\times K$ matrix with entries $[\Lambda ]_{ij}=\lambda _{ij}$ for $1\leq i \leq k$ and $i\leq j \leq k$.

\end{lemma}
\begin{proof}
We can define the Hadamard product of matrices $X$ and $Y$, denoted $X\circ Y$,  to be their elementwise product. That is, 
\[
[X\circ Y]_{ij} = [X]_{ij} \cdot [Y]_{ij},
\] 
where $[X]_{ij}$ is the $(i,j)$ entry of the matrix $X$.
Now, the Lagrangian function is
\[
\mathcal{L}(\Phi) =  \Tr  \left( \Phi^T W \Phi \right) + \mu || \Phi||_1 -  \sum_{i,j}\left[ \left( \Phi^TA \Phi- \Id \right) \circ \Lambda \right]_{ij} .
\]
From Lemma~5.1.5 of \cite{hornandjohnson} we have that $ \sum_{i,j}\left[ X\circ Y \right]_{ij}= \Tr (XY^T) $.
From this the result follows. 
\end{proof}
The method of Lagrangian multipliers requires the solution of $\dfrac{\partial \mathcal{L}}{\partial \Phi} =0 $. 
Using the facts that for matrices $X$, $Y$, $Z$ and function $f$, 
\begin{align*}
\dfrac{\partial }{\partial X } \Tr (XYX^TZ)&=ZXY+Z^TXY^T, \text{\quad (\cite{matrixcookbook} equation 118)}, \\
\dfrac{\partial }{\partial X^T } f(X) &= \left( \dfrac{\partial }{\partial X} f(X) \right) ^T ,
\end{align*}
we see that the equation $\dfrac{\partial \mathcal{L}}{\partial \Phi} =0 $ gives
\begin{align*}
2W\Phi+ \mu \sign (\Phi) - 2 A \Phi \Lambda &=0 \\
\Phi^T W \Phi+ \dfrac{\mu}{2} \Phi^T \sign (\Phi) &=\Phi^T A \Phi  \Lambda \\
\Phi^T W \Phi+ \dfrac{\mu}{2} \Phi^T \sign (\Phi) &= \Lambda .
\end{align*}
By considering the diagonal of both sides, and noting that for a vector $\varphi $ we have $\varphi^T \sign (\varphi) =||\varphi ||_1$, we can associate a number to each column of $\Phi $, i.e., to a compressed manifold mode. 
\begin{definition}
Let $\varphi $ be a compressed manifold mode of the operator $A^{-1}W$ with compression factor $\mu$. The {\textbf{compressed eigenvalue}}, $\lambda $, associated to $\varphi $ 
is the number
\[
\lambda = \varphi^T W \varphi+ \dfrac{\mu}{2} ||\varphi ||_1 .
\]
\end{definition}
Note that when $\mu =0$ and $K$ is the the number of rows of $W$ we get the standard eigenvalues of the operator $A^{-1}W$.

Figure~13 of \cite{CMM14} demonstrated that their algorithm consistently found similar modes for non-isometrically deformed meshes, in particular for a standing, crouching and limb bending man from the SCAPE dataset. This experiment was rerun with the accelerated algorithm (with $\mu =0.008$) and the modes were ordered and flipped (see Section~\ref{sec:flip} for an explanation of the latter). This is pictured in Figure~\ref{fig:cmm14fig13} and the resulting compressed eigenvalues are given in Table~\ref{tab:30menlambda}.

\begin{figure}
\begin{center}
\includegraphics[width=12cm]{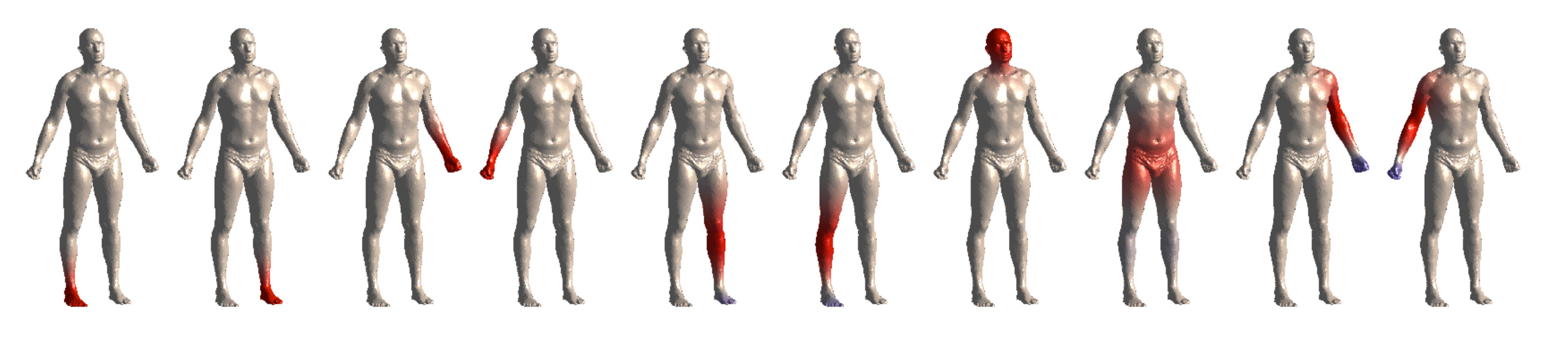}

\includegraphics[width=12cm]{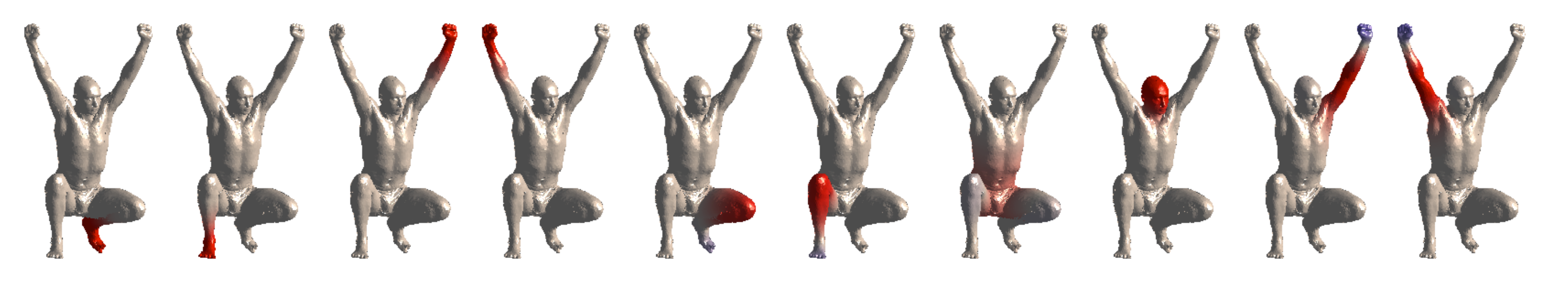}

\includegraphics[width=12cm]{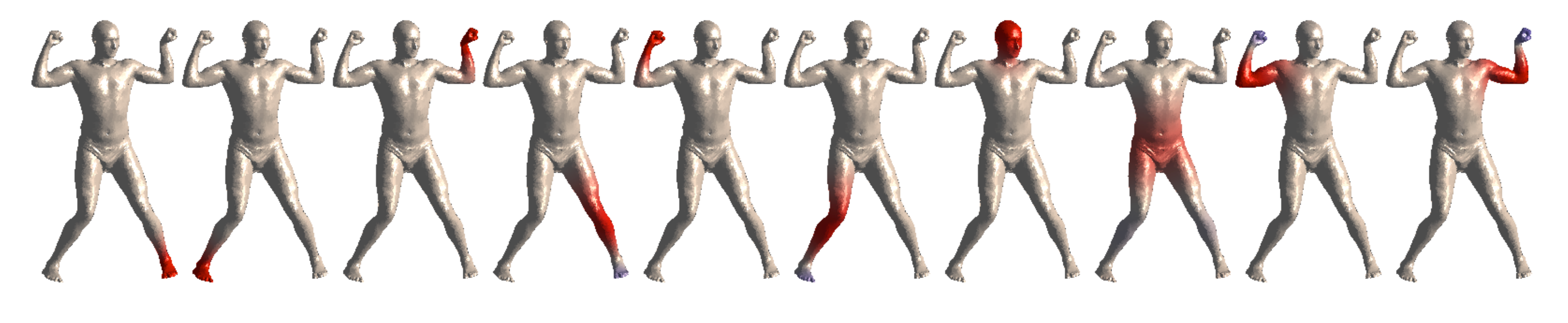}

\caption{\label{fig:cmm14fig13}
Ordering for three meshes, cf.~Figure~13 of \cite{CMM14}.}
\end{center}
\end{figure}

\begin{figure}
\begin{center}
\includegraphics[width=12cm]{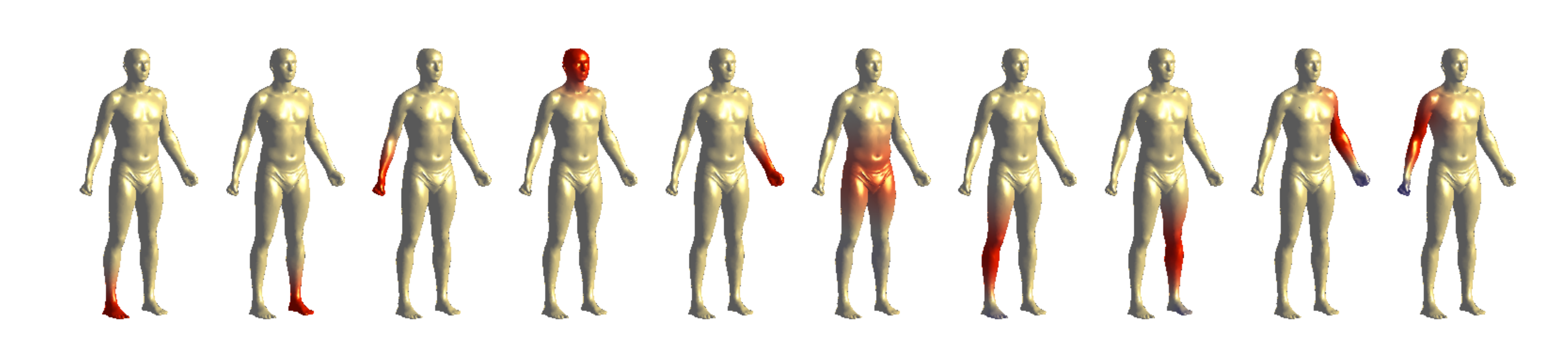}

\includegraphics[width=12cm]{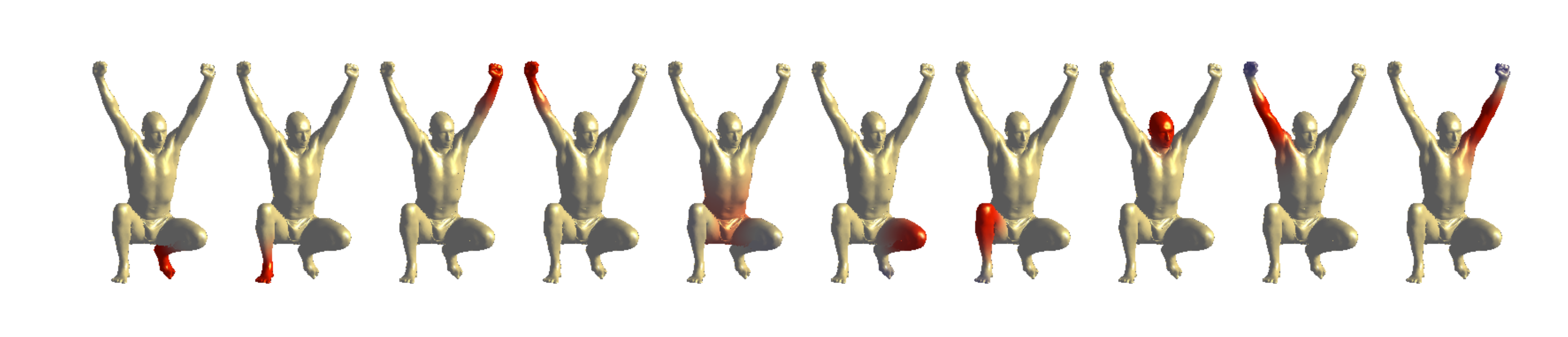}

\includegraphics[width=12cm]{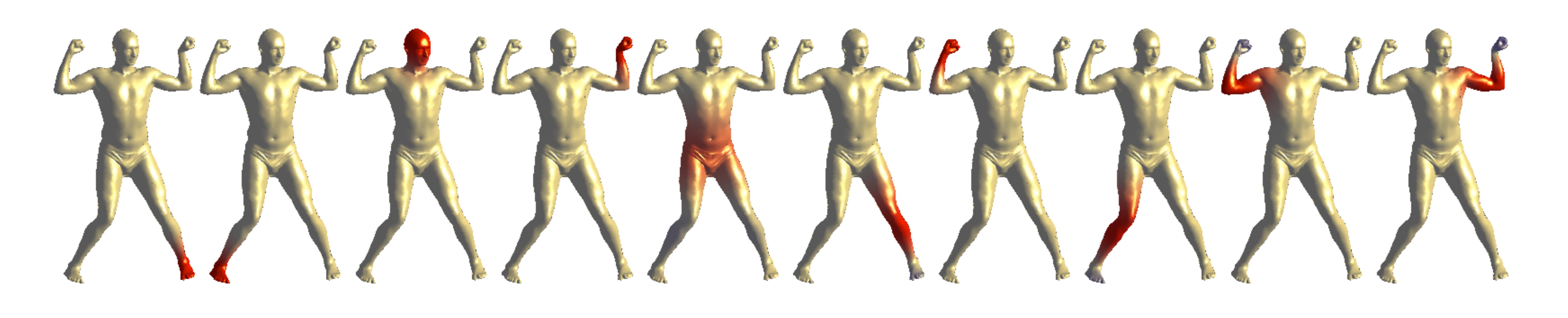}

\caption{\label{fig:dir-ordered}
Ordering for three meshes by only the Dirichlet energy, note that the ordering is not as good as Figure~\ref{fig:cmm14fig13}.}
\end{center}
\end{figure}

\begin{table}
\begin{center}
\begin{tabular}{|l|l|l|c|l|}
\hline
 & Stand & Crouch & Bent limbs & Teapot\\
  & & & &  $10^{-4} \times $ \\
 \hline
1&	19.6155	&	20.3833	&	20.5455	&	34.1979 \\
2&	19.8722	&	20.5321	&	20.6391	&	57.9528 \\
3&	21.8682	&	21.4825	&	24.6300	&	60.1676 \\
4&	23.7901	&	23.0605	&	26.7960	&	60.1705 \\
5&	27.3304	&	25.6135	&	27.9656	&	60.9936 \\
6&	28.1828	&	27.9401	&	29.0192	&	62.9510 \\
7&	28.7432	&	30.4192	&	29.0294	&	63.0666 \\
8&	34.6949	&	33.1840	&	33.6681	&	63.2446 \\
9&	39.4302	&	36.5946	&	38.3945	&	63.3030 \\
10&	39.5416	&	36.8892	&	38.4758	&	63.3360 \\
\hline
\end{tabular}			
\end{center}
\caption{\label{tab:30menlambda}Compressed eigenvalues for some meshes.}
\end{table}

The compressed eigenvalues can be closely grouped. For example, Aquarius, the Water Carrier, pictured in Figure~\ref{fig:locsup} has compressed eigenvalues
$2.2223$, $2.4737$, $2.5792$, $2.6721$, $2.8335$, and $3.4063$ for $K=6$ and $\mu=0.001$. (The accelerated algorithm took $3308$ iterations from a random initialization.)

For the human meshes we can see in Table~\ref{tab:30menlambda} that they range from $19.6155$ to $39.5416$. To put this into perspective, the first ten eigenvalues of the LBO range from $0$ to $31.7886$.
 One can see clearly in Figure~\ref{fig:cmm14fig13} that eight of the modes occur in obvious pairs: pair of feet, pair of legs, pair of hands, and pair of arms. This pairing is reflected in the compressed eigenvalues, for example for the Stand mesh the compressed modes supported on the feet correspond to compressed eigenvalues $19.6155$ and $19.8722$. This means that it is highly likely that, just as in the eigenvalue case of symmetric objects, numerical errors can cause compressed eigenvalues to be out of order. This is possibly what is happening with the human poses where the first compressed eigenvalue is sometimes the left foot, sometimes the right.

A set of twenty modes is shown in Figure~\ref{fig:teapot} for the classic  teapot. Here $\mu =0.0008$ and the initialisation was not random but the first twenty eigenfunctions of the discrete Laplace-Beltrami operator. Note again that the supports of the modes coincide with what a human might identify: the teapot spout, handle, lids, and knob of lid. 
The modes divide the spout in three. 
It would interesting to investigate how varying the parameter $\mu $ affects symmetry in the support of the modes.

Note again, that similar modes arising from symmetries have closely grouped compressed eigenvalues. This can be seen in Table~\ref{tab:30menlambda} where the 3rd to 5th compressed eigenvalues correspond to the three regions on the teapot lid. 
\begin{figure}
\begin{center}
\includegraphics[width=12cm]{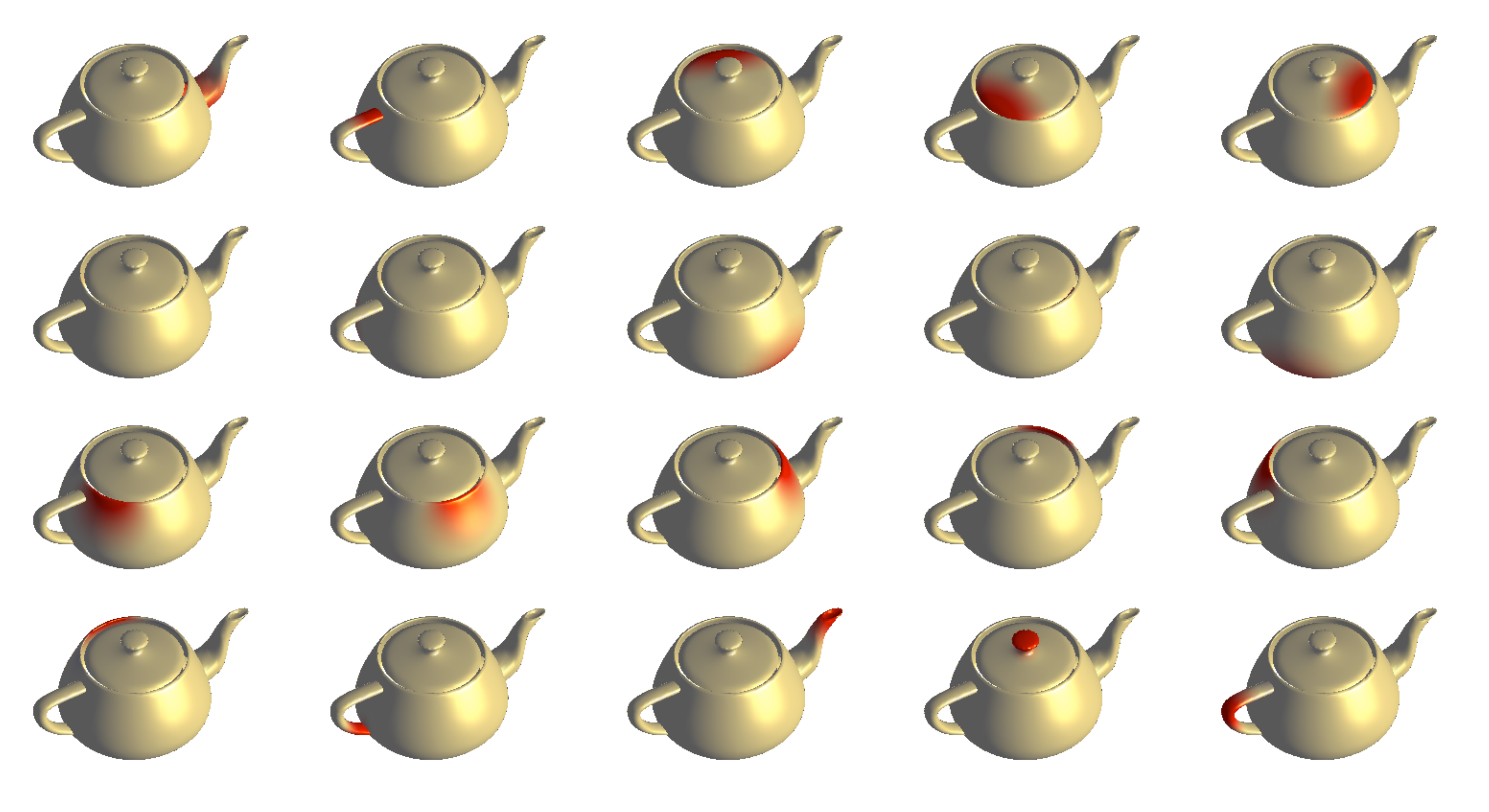}
\caption{\label{fig:teapot}Modes on the teapot.}
\end{center}
\end{figure}

\section{Accuracy and Consistency}
\label{sec:accuracy}
A measure of the accuracy of a numerical approximation of an eigenvector $v$ of a linear map $L$ can be calculated by $Lv-\lambda v$. The theory behind the calculation of ordering gives us an analogous measure of the accuracy of the algorithm in calculating modes. A mode $\varphi $ with compressed eigenvalue should satisfy the condition $W\varphi +(\mu/2)\sign (\varphi )=\lambda A\varphi $ and this can be used to check accuracy.

In both the original and accelerated versions of the algorithm the average error of entries of the vector 
$W\varphi +(\mu/2)\sign (\varphi )-\lambda A\varphi $
was of the order $10^{-3}$ to $10^{-4}$.  Hence, although we cannot have much confidence in the accuracy of the methods we can say that accuracy is not lost by using the faster algorithm. The accuracy aspect of CMMs requires further investigation. In particular, what is the trade-off between speed and accuracy and what is the relationship between the algorithm stopping tolerances $\epsilon^{\text{abs}}$ and $\epsilon^{\text{rel}}$ and accuracy.

One method to improve the accuracy was to perform 200-400 iterations of the accelerated algorithm, set all entries in $\Phi $ below a preset tolerance to zero and then restart the algorithm with this new $\Phi $. This did not conclusively improve accuracy but had the advantage of occasionally increasing the speed significantly. The results were inconsistent and further study is required. One could also try a similar approach by setting to zero the entries of $\Phi$ that correspond to large absolute entries of $W\varphi +(\mu/2)\sign (\varphi )-\lambda A\varphi $ for the different $\varphi $.

We now turn to consistency.
Given a random initialization it is clearly plausible that resulting modes are inconsistent, that is distinct runs of the algorithm with fixed $\mu$ and $K$, etc, may produce a different collection of modes. However, experiments show that for certain values of $\mu $ the calculation of modes by the accelerated method is very consistent. For the SCAPE dataset models (Stand, Crouch and Bent Limbs in this paper) repeated experiments on random initializations involving numbers from $0$ to $1$ produced less than 1\% inconsistent results for $\mu $ around $0.008$. 
For other models, for example Aquarius, it was harder to locate a good value of $\mu $ to use. 

The value of $\mu $ is important for applications as $\mu $ determines the `size' of the support of modes, that is the amount the support covers the model. If $\mu $ is large, then the supported area is small.  
If the supported area is small, then it is likely that supported areas do not interact during the iterations and this can in theory lead to inconsistency. 
Hence, the value of $\mu $ is important for consistency but is poorly understood.

One way of avoiding inconsistency is to use the eigenfunctions of the Laplace-Beltrami operator as the initialization of $\Phi $. This need not give the same $K$ modes as the random initialization and it would be interesting to investigate further the use of this initialization as it worked quite well in some experiments. For example, the teapot in Figure~\ref{fig:teapot} was calculated using the first twenty eigenfunctions and this was very consistent for small changes of $\mu$ in the sense we had the same modes in the same order.
Nonetheless, what are the quantitative changes to the modes when $\mu $ is varied is still an open question.

When one changes the number $K$ of modes to be calculated, then one gets inconsistent results. For example, compare Figure~\ref{fig:choiceK} with Figure~\ref{fig:cmm14fig13} where $K$ is $8$ and $10$ respectively. The modes calculated for $K=8$ are not the first $8$ modes calculated for $K=10$. However, Figure~\ref{fig:cmm14fig13} gives one confidence that good consistency is achievable even between near isometric shapes.
\begin{figure}
\begin{center}
\includegraphics[width=12cm]{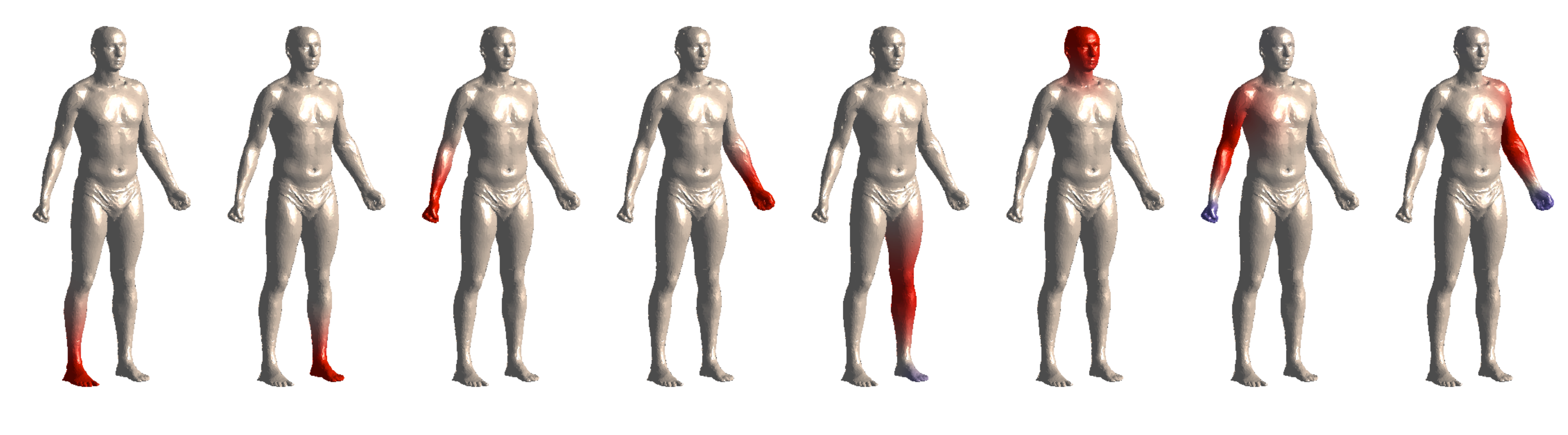}
\caption{\label{fig:choiceK}Effect of the choice of $K$.}
\end{center}
\end{figure}

\section{Orientation and flipping of modes}
\label{sec:flip}
A disadvantage of the Laplacian unit eigenfunctions is that they are defined only up to sign and there is no obvious way to `flip' them so that eigenfunctions between different manifolds can be directly compared. For these eigenfunctions one expects an `equal amount' of positive and negative values. That is, the integral of  the positive values equals the integral of the negative values. (This is because the eigenfunctions are orthogonal to the constant vector, the eigenfunction for the zero eigenvalue.)
However, for general compressed modes there is not this balancing of the positive and negative. Hence we can flip modes so that the ambiguity of parity is removed.. 

In practice, the values of a mode are dominated by numbers of either positive or negative values. For a mode $\varphi $ the number $\sign \left( \max(\varphi)+\min(\varphi) \right)$ will, in practice, be $\pm 1$ depending on the sign of the dominant values. One can then use this to change the sign of the columns of $\Phi $. In this way, the modes chosen to be dominated by positive values. A problem only occurs when $\max (\varphi)=-\min (\varphi )$ but this is unlikely in practice.

An alternative and perhaps more theoretically sound way to flip the mode is to integrate it over the manifold and define its sign to be the sign of the resulting value. 

That the positive values dominate after flipping can be seen in the figures where red is positive.

\section{Lumped and unlumped mass matrices}
\label{sec:lump}
The mass matrix $A$ in $L=A^{-1}W$ contains information about the area around the vertices. Many variations are possible, for example one-third area, Voronoi area or uniform area, see \cite{LevyZhang}. 
Given such a matrix one can `lump' the matrix by summing all the elements of a row and putting the result in the diagonal. A mass matrix with non-zero entries only on the diagonal is called a {\bf{lumped matrix}}. 

The algorithm in \cite{CMM14} is given for lumped matrices. In that paper $D$ is a diagonal matrix and it is easy to calculate $D^{1/2}$ and its inverse $D^{-1/2}$, where $D^{1/2}$ is formed from the square roots of the diagonal entries. For a positive definite matrix $B$ such as the mass matrix there exists a square root matrix $B^{1/2}$ such that $B^{1/2}B^{1/2}=B$. However, the calculation of this can be computationally expensive and hence though it is theoretically possible with this method to include unlumped mass matrices in the algorithm of \cite{CMM14} this may seriously impact the speed of computation.

Nonetheless, unlumped matrices can be used. 
In \cite{CMM14} the first part of the algorithm involves letting $Y=D^{1/2} (S-U_S+E+U_E)$, 
finding an SVD factorization of $Y^TY=VWV^T$ and then the closed form of the minimization problem is
\[
\Phi = D^{-1/2} YVW^{1/2}V^T .
\]
If instead we take $\widetilde{Y}=S-U_S+E+U_E$, and so $Y=D^{1/2}\widetilde{Y}$ then $Y^TY=VWV^T$  simply becomes, 
$\widetilde{Y}^TD\widetilde{Y}=VWV^T$ and the closed form of the minimization problem is then
\[
\Phi = D^{-1/2} \left( D^{1/2}\widetilde{Y} \right) VW^{1/2}V^T 
=\widetilde{Y} VW^{1/2}V^T .
\]
Thus with this modification to the algorithm we can avoid the expensive computation of the square root of the mass matrix and allow use of unlumped matrices. 

Numerous experiments with different meshes, values of $K$ and $\mu$, failed to show that either the lumped or unlumped version was superior in terms of speed or accuracy. It seems unlikely that lumping does not at least have some effect and so while the experiments were inconclusive it may be the case that there do exist identifiable situations in which one method is superior to the other.

To ensure consistency and comparability the experiments described in this paper were done with lumped matrices. It should be noted that lumped matrices are used for the LBO much more commonly than unlumped. The aim here is to show that we can use unlumped if we wish.

\section{Conclusions and future work}
Compressed manifold modes were demonstrated in \cite{CMM14} to be robust with respect to noise and holes and have the potential for use in geometry processing applications. Indeed as they generalize to surfaces the compressed modes introduced in \cite{ozolicnvs2013compressed} they also have potential for application in solving PDEs on surfaces. In this paper it has been shown that the algorithm of \cite{CMM14} for the computation of compressed manifold modes can be considerably improved - by almost a factor of 2. The correct method for ordering modes has been demonstrated (with proof) and the modes have been oriented.

The study of the numerical calculation of eigenvalues has had the advantage of decades of work and the development of many optimization techniques and so unsurprisingly the calculation of compressed manifold modes is not  competitive in terms of speed. This paper helps reduce this uncompetitiveness but the development of further improvements would be welcome.

Compressed manifold modes are of interest because they are functions with local support. A crucial insight is that if one is looking for a collection of functions to form a basis of functions on the manifold, then speed and accuracy are not that important, it is the orthonormality of the collection that is important. In the case of the CMM algorithms the locality of support and orthonormality can, and in the experiments usually do, appear rapidly, say within 400-600 iterations of the algorithm. Therefore it would be good to analyse how quickly these properties are achieved.

Further work of interest includes the following:
\begin{enumerate}
\item Investigate and define for numerical calculations what it means for the modes to be locally supported and orthonormal. For the latter we need to check that $\Phi ^T A \Phi $ is close to the identity. How close for practical applications?

\item The random initialization can lead to a wide variation in computation time. For example, in the 30 calculations used in Table~\ref{tab:compare}, the Bent Limbs mesh with $K=10$, $\mu =0.008$ took between $736$ and $4323$ iterations, the latter is a factor of nearly $6$ greater than the former. 
Can we do better or at least be more consistent?

\item Can inconsistency of the collection of modes calculated be reduced by restricting the band of random numbers used as an initialization? See Section~\ref{sec:accuracy}.

\item Can accuracy of the algorithms be improved by some method, for example truncating the support as in Section~\ref{sec:accuracy}.

\item How do changes in the compression factor $\mu $ affect the consistency and accuracy of the algorithms? 
\end{enumerate}

%
%
\bibliographystyle{plain}
\bibliography{cmm-arxiv.bbl}

\end{document}